\documentclass[11pt,a4paper]{article}
\usepackage{authblk}
\usepackage[utf8x]{inputenc}
\usepackage{amsmath, amsthm, mathtools}
\usepackage{graphicx}
\usepackage{amsfonts}
\usepackage[hidelinks]{hyperref}
\usepackage{doi}
\usepackage{amssymb, natbib}
\usepackage{epstopdf}
\usepackage[disable]{todonotes}

\newcommand{\Circ}{\operatorname{circ}}
\newcommand{\R}{\mathbb{R}}
\newcommand{\C}{\mathbb{C}}
\newcommand{\bb}[1]{\boldsymbol{#1}}

\newcommand{\T}{^\top}
\newcommand{\N}{\mathbb{N}}
\newcommand{\tridiag}{\operatorname{tridiag}}
\newcommand{\diag}{\operatorname{diag}}
\newcommand{\I}{\imath}

\usepackage[english]{babel}
\usepackage{a4wide}
\newtheorem{theorem}{Theorem}[]

\newtheorem{lemma}[theorem]{Lemma}
\theoremstyle{definition}
\newtheorem{remark}[theorem]{Remark}
\newtheorem{example}[theorem]{Example}
\setcitestyle{numbers,square,comma}
\renewcommand{\baselinestretch}{1.2}
\begin{document}
\title{Valid parameter space of a bivariate Gaussian Markov random field with a generalized block-Toeplitz precision matrix}
\date{}
\author[1]{Mattia Molinaro\thanks{mattia.molinaro@math.uzh.ch}}
\author[1,2]{Reinhard Furrer\thanks{reinhard.furrer@math.uzh.ch}}
\affil[1]{Institute of Mathematics, University of Zurich, Winterthurerstrasse 190, Zurich, Switzerland}
\affil[2]{Institute for Computational Science, University of Zurich, Winterthurerstrasse 190, Zurich, Switzerland}
%\dfrac{d}{•}\date{}
\maketitle
\begin{abstract}
Gaussian Markov random fields (GMRFs) are extensively used in statistics to model area-based data and usually depend on several parameters in order to capture complex spatial correlations. In this context, it is important to determine the valid parameter space, namely the domain ensuring (semi) positive-definiteness of the precision matrix. Depending on the structure of the latter, this task can be challenging. While univariate GMRFs with block-Toeplitz precision are well studied in the literature, not much is analytically known about bivariate GMRFs. So far, only restrictive sufficient conditions and brute-force approaches were proposed, which are computationally expensive for the size of modern datasets. In this paper, we consider a bivariate GMRF, which is part of a hierarchical model used in spatial statistics to analyze data coming from projections of regional climate change. By extending classical convergence results of univariate fields with toroidal boundary conditions to fields without boundary conditions, we provide asymptotically closed-form expressions of the valid parameter space. We develop a general methodology that can be used to determine the valid parameter space of bivariate \mbox{GMRFs} whose precision matrix has a generalized block-Toeplitz structure and for which classical convergence results are not directly applicable. Finally, we quantify the rate of convergence of our approach through a numerical study in \texttt{R}.
\bigskip

\bigskip

% 15A18: Eigenvalues, singular values, and eigenvectors
% OR 15A42: Inequalities involving eigenvalues and eigenvectors
% 62M05: Markov processes: estimation
% OR 62M40: Random fields; image analysis
% OR 62M30: Spatial processes

\bigskip
\end{abstract}

\renewcommand{\baselinestretch}{1.15}

\section{Introduction}
The analysis of large datasets with complicated spatio-temporal dependence structures is among the major challenges in modern statistics. With respect to this, sophisticated models are used in order to capture such dependencies. In \cite{opac-b1120661, panoramicaGenton}, an overview of these models is provided. In addition, computational efficiency is crucial, since the size of the datasets available to the scientific community has been growing steadily over the years. 

A broad class of the models currently in use in spatial statistics are Gaussian latent models, discussed in full detail in \cite{rue2009}. The case of observations on a lattice is of particular interest. Several case studies are discussed, for instance in \cite{GMRFbook}, and some related modeling and computational challenges are covered in \cite{paperFlorian}. In the case of Gaussian observations $\bb{y}$ over a (regular) lattice, a prototypical model is
\begin{subequations}\label{modellogenerale}
\begin{align}
\bb{y}\mid\bb{\beta},\bb{z},\sigma & \sim  \mathcal{N}(\bb{x}\T\bb{\beta}+\bb{z}, \sigma^2\bb{I}),
\label{primo layer}
\\
\bb{z}\mid\bb{\theta} & \sim  \mathcal{N}(\bb{0},\bb{Q}^{-1}(\bb{\theta})),
\label{secondo layer}
\end{align}
\end{subequations}
where $\bb{x}$ are known covariates. In the Bayesian setting (see, e.g., \citep{gelman2013bayesian}), prior distributions on the (hyper) parameters $\bb{\beta},\bb{\theta}$, and $\sigma$ are to be specified. In the previous model, the layer (\ref{primo layer}) accounts for the fact that the data $\bb{y}$ is assumed to be independent at different locations on the lattice, conditionally on a linear term $\bb{x}\T\bb{\beta}$ and a latent field $\bb{z}$. The layer (\ref{secondo layer}) models complicated spatial dependencies through the parameter $\bb{\theta}$. In this context, the spatial correlation structure can be specified through the off-diagonal non-zero elements of the precision matrix $\bb{Q}(\bb{\theta})$ of the random vector $\bb{z}$ (see, e.g., \cite{GMRFbook} for details). Equivalently, the model (\ref{modellogenerale}) can be specified in terms of the expected value of the Gaussian full-conditional distribution at each location, in the framework of a conditional autoregressive (CAR) model. The CAR models and their importance in spatial statistics are well-known and discussed in \cite{Mardia1988265}.

The ease of interpretability of the model (\ref{modellogenerale}) comes at a price. In the applications, at least three crucial aspects are to be considered. Depending on the structure and the size of the precision matrix $\bb{Q}(\bb{\theta})$, it can be computationally intensive to evaluate quadratic forms and log-determinants involving such matrices. Of course, this can lead to high computational costs when, for instance, implementing a Markov Chain Monte Carlo (MCMC) sampler for large datasets (see, e.g., \citep{gamerman1997markov}). It is then necessary to develop efficient computational strategies in order to evaluate the above-mentioned quantities. In the literature, the case of a univariate field over a regular lattice of size $n = n_1\times n_2$ with a block-Toeplitz precision is thoroughly covered. The main idea, discussed in \cite{Mardia1988265}, is to consider a lattice with toroidal boundary conditions. Under this approximation, the precision matrix becomes block-circulant and very efficient algorithms for the latter class of matrices can be used. They are based on the (multidimensional) fast Fourier transform (FFT), discussed in full detail in \cite{libroFFT}. The third aspect that has to be considered is to determine the valid parameter space for the model~(\ref{modellogenerale}), namely the set of parameters for which the precision matrix $\bb{Q}(\bb{\theta})$ is (semi) positive-definite. The importance of this task is pointed out, for instance in \cite{spam, GMRFbook}. Once again, the solution of this problem is essential if we want to implement a MCMC sampler or an optimization algorithm for the maximum likelihood estimation for a model of the form (\ref{modellogenerale}). In the former case, we need to know the set over which we can impose a prior distribution on $\bb{\theta}$, whereas in the latter case the domain of the function to be optimized must be determined. To our knowledge, in the literature, only specific precision matrix structures (block-Toeplitz matrices, discussed in \cite{GMRFbook}, and block-circulant, e.g., in \cite{Xu12abayesian}), and univariate fields (e.g., in~\cite{paperunivariato}) are covered. The case of bivariate fields, on the contrary, was not analytically explored. So far, only approaches aimed at determining analytically tractable subsets of the domain, such as the diagonal dominance criterion, were considered, for example in \cite{GMRFbook}. The obtained results were, however, unsatisfactory in terms of coverage of the whole valid parameter space. Apart from this, the determination of the parameter space relied on brute-force \lq\lq trial-and-error\rq\rq\ approaches, based on the fact that a Hermitian matrix is positive-definite if and only if it admits a unique Cholesky factorization (see \cite{spam} for the details). However, these techniques are computationally expensive.

In this paper, we consider the model introduced in \cite{Sain:Furr:Cres:11} to analyze bivariate data coming from several regional climate models (RCMs). In particular, we focus on the associated bivariate GMRF, which is in the form of the layer (\ref{secondo layer}).
It depends on seven real parameters, namely $\phi, \rho_{11}, \rho_{12}, \rho_{21}, \rho_{22}, \tau_1$, and $\tau_2$. These in turn stem from the layer decomposition introduced in \cite{Sain:Furr:Cres:11}, which is an approach to interpret a multivariate GMRF as an instance of univariate GMRF. The main idea behind this concept is to work with one layer for each variable of interest. In this context, $\phi$ corresponds to within-location variability, whereas $\rho_{11}$ and $\rho_{22}$ describe within-variable variability. Furthermore, $\rho_{12}$ and $\rho_{21}$ account for the cross-location variability. These parameters correspond to the coefficients of a linear combination through which the expected values of the aforementioned full-conditional distributions are defined. Finally, $\tau_1^2$ and $\tau_2^2$ are the marginal variances of these bivariate distributions. Due to the different nature of these parameters, it natural to set the following:
\begin{gather*}
\bb{\theta} = (\phi,\rho_{11},\rho_{12},\rho_{21},\rho_{22})\T,\\
\bb{\tau} = (\tau_1,\tau_2)\T,
\end{gather*}
which induce the precision matrix $\bb{Q}_{n_1, n_2}(\bb{\theta},\bb{\tau})\in\R^{2n\times 2n}$. 

The main goal of this paper is to provide asymptotically closed-form expressions of the valid parameter space, namely
\begin{equation}
\label{insiemevalidita}
\bb{\Theta}_{n_1, n_2} = \left\{(\bb{\theta},\bb{\tau})\in\R^5\times\R^2\mid\bb{Q}_{n_1, n_2}(\bb{\theta},\bb{\tau})\succ \bb{0}\right\}.
\end{equation}
In the previous definition, we assumed strictly positive-definite precision matrices. Although in this paper we will mainly focus on a specific bivariate GMRF, one important goal will be to provide a methodology that can be extended to multivariate GMRFs, whose precision matrix has a generalized Toeplitz matrix for which no analytic results for the eigenvalues can be obtained. In other words, the aim of our methodology is to overcome the aforementioned lack of theoretical results for multivariate GMRFs. 

The structure of the paper follows. In Section \ref{The Model and Preliminary Results}, we will briefly describe the model outlined in \cite{Sain:Furr:Cres:11} and enumerate some preliminary results, which will be necessary later. In addition, extending the work presented in \cite{CIT-006}, we will develop a toroidal boundary condition approximation: $\widetilde{\bb{Q}}_{n_1,n_2}(\bb{\theta},\bb{\tau})$ of $\bb{Q}_{n_1,n_2}(\bb{\theta},\bb{\tau})$ in order to characterize the domain (\ref{insiemevalidita}) through
\begin{equation}
\label{insiemevaliditaperturbato}
\widetilde{\bb{\Theta}}_{n_1, n_2} = \big\{(\bb{\theta},\bb{\tau})\in\R^5\times\R^2\mid\widetilde{\bb{Q}}_{n_1, n_2}(\bb{\theta},\bb{\tau})\succ \bb{0}\big\}.
\end{equation}
The asymmetry in the parameters $\rho_{12}$ and $\rho_{21}$, which was introduced in  \cite{Sain2007226, Sain:Furr:Cres:11} in order to better describe spatial dependence, for instance in the aforementioned RCM data, will play a central role. The importance and feasibility of our approach in the applications will be further discussed. In Section~\ref{The Main Result}, we will state and formally prove our main result, namely the convergence—in a suitable sense—of $\widetilde{\bb{\Theta}}_{n_1, n_2}$ to $\bb{\Theta}_{n_1, n_2}$. This result will highlight the fact that considering a regular grid with toroidal boundary conditions is not only useful to evaluate log-determinants and quadratic forms but is also useful, for instance, to efficiently sample a prior distribution defined over the set~(\ref{insiemevalidita}). This in turn improves what is discussed in \cite{GMRFbook}, where a regular grid without boundary conditions is embedded in a bigger grid with toroidal boundary conditions. The limitations of other known convergence results for this problem (e.g., the weak convergence of sequences of matrices discussed in \cite{CIT-006}, the Szeg\"{o} theorem covered in \cite{opac-b1092435}, etc.) will also be discussed. In Section \ref{Assessment of the Rate of Convergence}, we will implement a thorough simulation study aimed at numerically assessing the rate of convergence of our approximation. We will additionally discuss some related aspects, which are important in the applications. Finally, in Section \ref{Discussion and Outlook}, we will provide some conclusive remarks.

\section{The Model and Preliminary Results}
\label{The Model and Preliminary Results}
We briefly outline the model introduced in \cite{Sain:Furr:Cres:11} in order to analyze bivariate data over a regular lattice of size $n = n_1\times n_2$, with particular emphasis on the structure of the precision matrix of the latent field $\bb{z}$ of equation (\ref{secondo layer}), which will play a crucial role in what follows. Up to a permutation of the rows and columns, the precision matrix of size $2n\times 2n$ can be rewritten as follows:
\begin{equation}
\label{matriceprecisione}
\bb{Q}_{n_1, n_2}(\bb{\theta},\bb{\tau}) = 
\begin{pmatrix}
\bb{\tau}_1^{-1} & \\
 & \bb{\tau}_2^{-1}
\end{pmatrix}
\begin{pmatrix}
\bb{T}_{n_1,n_2}(\rho_{11}, 1, \rho_{11}) & \bb{T}_{n_1,n_2}(\rho_{21}, \phi, \rho_{12}) \\
\bb{T}_{n_1,n_2}\T(\rho_{21}, \phi, \rho_{12}) & \bb{T}_{n_1,n_2}(\rho_{22}, 1, \rho_{22})
\end{pmatrix}
\begin{pmatrix}
\bb{\tau}_1^{-1} & \\
 & \bb{\tau}_2^{-1}
\end{pmatrix},
\end{equation}
where $\bb{\tau}_1 = \diag(\tau_1)\in\R^{n\times n}$, $\bb{\tau}_2 = \diag(\tau_2)\in\R^{n\times n}$ and $\diag(x)=x\,\bb{I}$.
The four blocks of the form $\bb{T}_{n_1,n_2}(\cdot)$ in the previous equation have size $n\times n$ and are block-Toeplitz, with the following structure:
\begin{equation}
\label{funzioneT}
\begin{aligned}
\bb{T}_{n_1,n_2}(x,y,z) & =  
\begin{pmatrix}
\tridiag(x,y,z) & \diag(z) &  &  &  \\
\diag(x) & \tridiag(x,y,z) & \diag(z) &  & \\
 & \diag(x) & \ddots & \ddots  &  \\
 &  & \ddots & \ddots & \diag(z)\\
 &  &  & \diag(x) &  \tridiag(x,y,z)
\end{pmatrix},
\end{aligned}
\end{equation}
where
\begin{equation*}
\begin{aligned}
\tridiag(x,y,z) & =  
\begin{pmatrix}
y & z &   &   &  \\
x & y & z &   &  \\
  & x & \ddots & \ddots & \\
  &   & \ddots & \ddots & z \\
  &   &        & x & y
\end{pmatrix}\in\R^{n_1\times n_1}.
\end{aligned}
\end{equation*}
Then, $\bb{T}_{n_1,n_2}(\rho_{11}, 1, \rho_{11})=\bb{T}_{n_1,n_2}(\rho_{22}, 1, \rho_{22})$ if and only if $\rho_{11}= \rho_{22}$. This means that $\bb{Q}_{n_1, n_2}(\bb{\theta},\bb{\tau})$ is in general not block-Toeplitz. This structure generalizes what is discussed in \cite{Kent1996379, GMRFbook}, where the precision (or, in turn, the variance-covariance) matrix is block-Toeplitz. Our precision matrix is also sparse, and the non-zero entries pattern is shown in the left panel of Figure~\ref{precisionepiccolaesempio}. It can be shown that $\bb{Q}_{n_1,n_2}(\bb{\theta},\bb{\tau})$ has at most $20n_1 n_2 - 8n_1 - 8n_2\in\Theta(n)$ non-zero entries. In this paper, we use the big theta notation, namely $f(n)\in\Theta(g(n))$ if and only if $f(n)$ is eventually bounded from both above and below by $g(n)$.  
\begin{figure}
\centering
\includegraphics[scale = 0.44, trim={0 2.6cm 0 3.cm},clip]{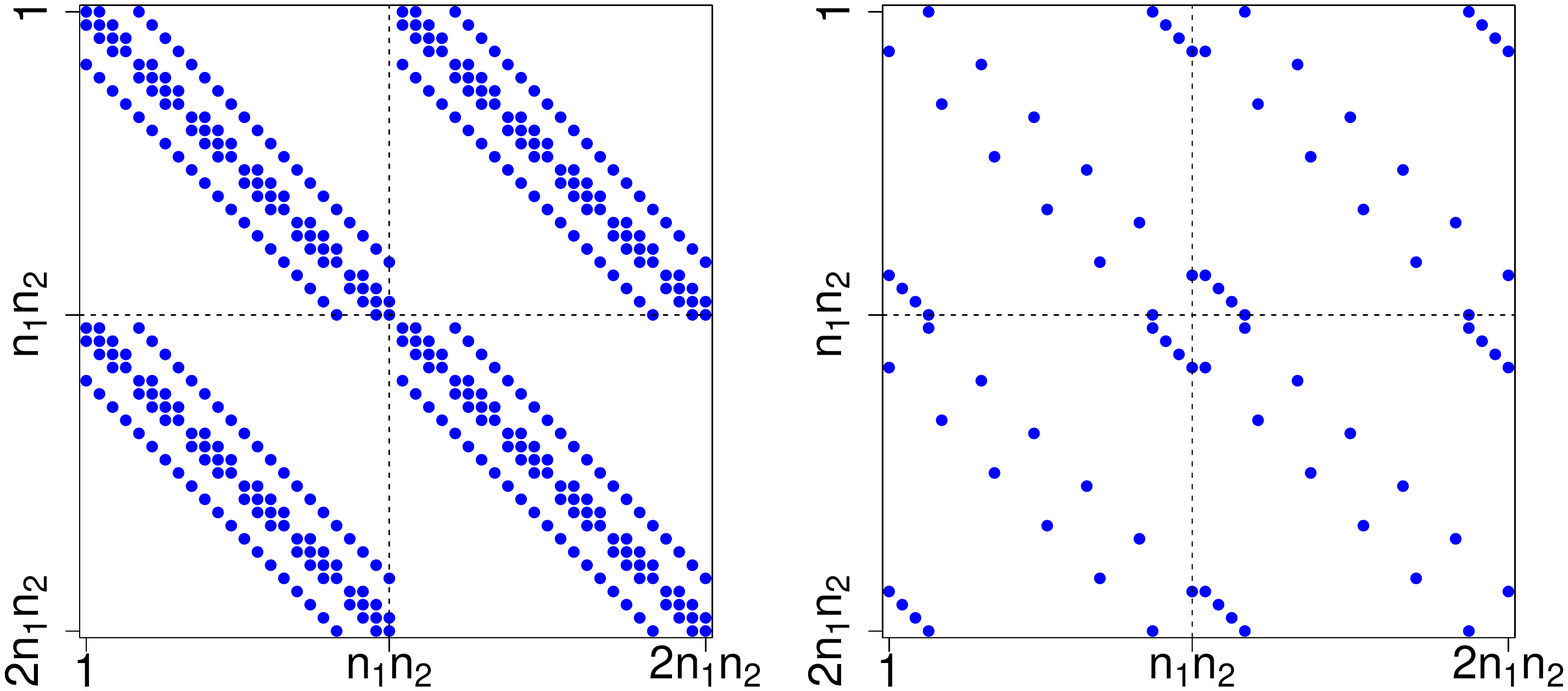}
\caption{The left and right panel show the sparsity pattern of the precision matrix (\ref{matriceprecisione}) and of the perturbation matrix $\Delta\bb{Q}_{n_1,n_2}(\bb{\theta})$, respectively. The blue points correspond to the non-zero entries of these matrices. In the display, $n_1 = 4$ and $n_2 = 6$.}  
\label{precisionepiccolaesempio}
\end{figure}

Recall that the first goal of this study is to characterize the valid parameter space~(\ref{insiemevalidita}). In the literature (e.g., in \cite{GMRFbook}) the diagonal dominance criterion is used in order to determine analytically tractable subsets. More precisely, it is well-known that an Hermitian matrix $\bb{A}\in\C^{n\times n}$ such that $\big\vert(\bb{A})_i^i\big\vert\geq\sum_{j\neq i} \big\vert(\bb{A})_i^j\big\vert$, for all $i = 1,2,\dots,n$, is semipositive-definite (a formal proof of this statement is provided in \cite{libroalgebralineare}). However, the converse does not hold in general. In order to highlight the importance of the convergence result that we will discuss in the next section, we will preliminarily deal with the limitations of the diagonal dominance criterion for the precision matrix (\ref{matriceprecisione}). With respect to this, we uniformly drew 500,000 values of $\bb{\theta}$ belonging to the valid parameter space (\ref{insiemevalidita}) with $n_1 = n_2 = 100$. Among these, we determined the ones also satisfying the diagonal dominance criterion. The ratio between the latter and the former was 0.1288. This rate of coverage is clearly unsatisfactory and agrees with \cite{GMRFbook} for univariate GMRFs. In addition, in order to attain a partial understanding of the geometrical structure of the set (\ref{insiemevalidita}), we additionally drew uniformly four values of $(\phi,\rho_{11},\rho_{22})\T$ from the set (\ref{insiemevalidita}). For each of them, we uniformly sampled 10,000 values of $(\rho_{12},\rho_{21})\T$ such that $(\phi,\rho_{11},\rho_{22}, \rho_{12},\rho_{21})\T$ belong to (\ref{insiemevalidita}). We chose this approach since it is relatively easy to devise closed-form bounds for each of the parameters $\phi,\rho_{11}$, and $\rho_{22}$, whereas the same does not hold for $\rho_{12}$ and $\rho_{21}$. The obtained results are displayed in Figure \ref{figuradd}. The blue points correspond to values of the parameters belonging to the set (\ref{insiemevalidita}). The ones also satisfying the diagonal dominance criterion are highlighted in orange.
Several features are striking. First, there are values of $(\phi,\rho_{11},\rho_{22})\T$ for which the diagonal dominance criterion provides no coverage. Even when it does, the coverage rate is clearly unsatisfactory. Second, the region determined by this criterion is symmetric around the origin, while the region determined by the valid points is generally not symmetric around this point. In other words, the diagonal dominance criterion is not able to capture the geometrical structure of the valid parameter space (\ref{insiemevalidita}).
\begin{figure}
\centering
\includegraphics[scale = 0.58]{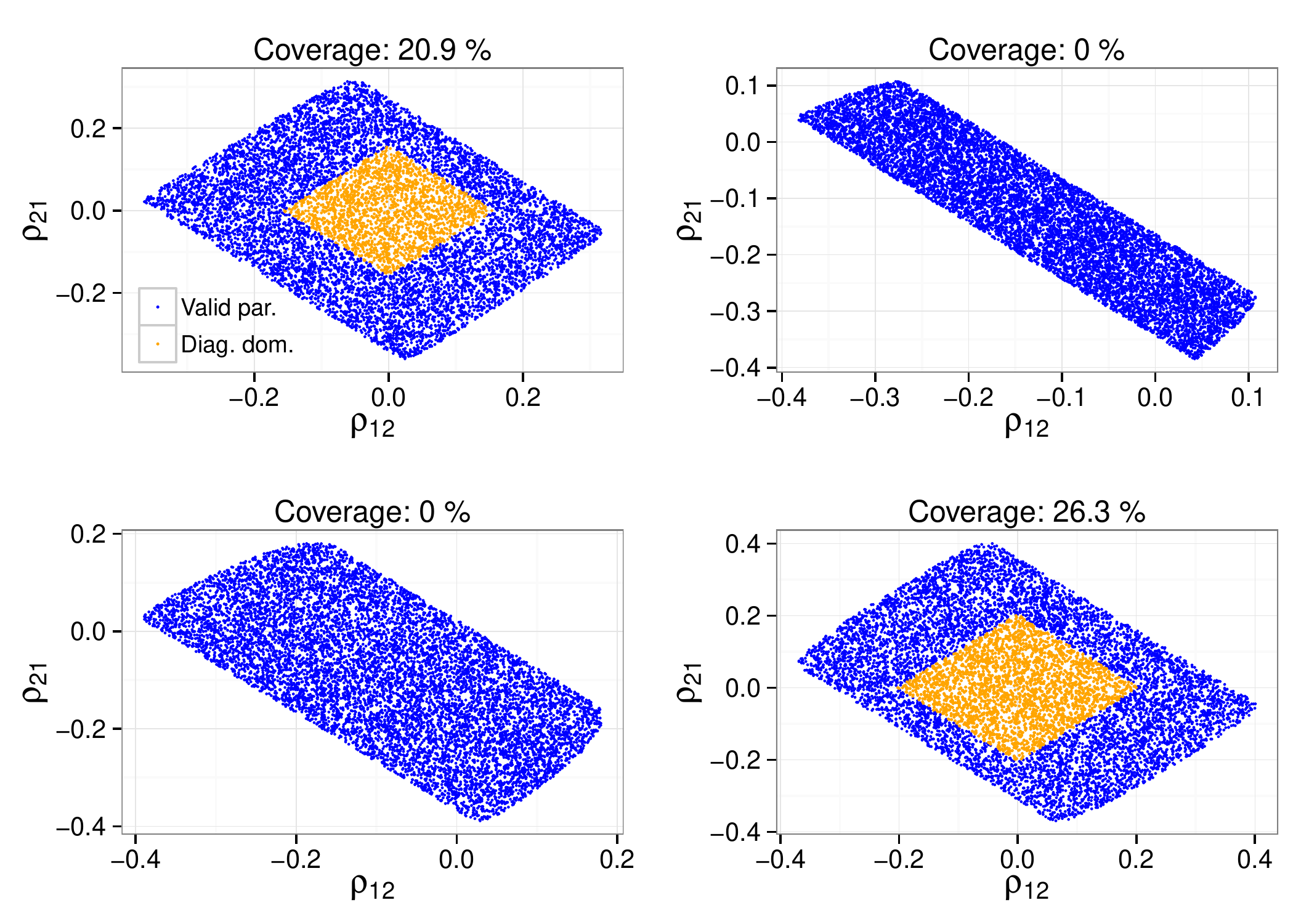}
\caption{in each panel, 10,000 values of $(\rho_{12},\rho_{21})\T$ were uniformly drawn from the valid parameter space (\ref{insiemevalidita}) conditional on a value of $(\phi,\rho_{11},\rho_{22})\T$. The orange points correspond to values of $\bb{\theta}$ also fulfilling the diagonal dominance criterion. The ratio between the orange and blue points is displayed on top of each panel.}
\label{figuradd}
\end{figure}

As a consequence, we need to resort to a different approach to describe the valid parameter space (\ref{insiemevalidita}). The latter problem is equivalent to determining the values of $\bb{\theta}$ and $\bb{\tau}$ such that the minimum eigenvalue of $\bb{Q}_{n_1, n_2}(\bb{\theta},\bb{\tau})$ is strictly positive. Consequently, we will focus on the spectrum of $\bb{Q}_{n_1, n_2}(\bb{\theta}, \bb{\tau})$, which will be denoted with $\sigma(\bb{Q}_{n_1, n_2}(\bb{\theta}, \bb{\tau}))$. In addition, we will adhere to the convention of enumerating the eigenvalues in decreasing order, namely $\lambda_1(\bb{Q}_{n_1, n_2}(\bb{\theta}, \bb{\tau}))\geq \lambda_2(\bb{Q}_{n_1, n_2}(\bb{\theta}, \bb{\tau}))\geq\dots\geq \lambda_{2n}(\bb{Q}_{n_1, n_2}(\bb{\theta}, \bb{\tau}))$. From equation (\ref{matriceprecisione}), it is clear that $\bb{Q}_{n_1, n_2}(\bb{\theta},\bb{\tau})$ is positive-definite only if $\tau_1,\tau_2 > 0$. Henceforth, we will only consider the inner matrix of equation (\ref{matriceprecisione}), which depends only on $\bb{\theta}$. Moreover, we will drop the dependence on $\bb{\tau}$ in order to keep the notation as straightforward as possible.

No closed-form formulas for the eigenvalues of $\bb{Q}_{n_1, n_2}(\bb{\theta})$ are available. More specifically, only results for banded block-Toeplitz matrices are discussed (see, e.g., \cite{blocktoeplitzspectra}). In the special case $\rho_{12} = \rho_{21}$, it can be shown that there exists an orthogonal matrix $\bb{V}_{n_1,n_2}\in\R^{n\times n}$ not depending on $\bb{\theta}$ such that
\begin{equation}
\label{caso inutile}
\begin{pmatrix}
\bb{V}_{n_1,n_2}\T & \\
 & \bb{V}_{n_1,n_2}\T
\end{pmatrix}\,
\bb{Q}_{n_1,n_2}(\bb{\theta})\,
\begin{pmatrix}
\bb{V}_{n_1,n_2} & \\
 & \bb{V}_{n_1,n_2}
\end{pmatrix} = 
\begin{pmatrix}
\bb{\Lambda}_{11}(\bb{\theta}) & \bb{\Lambda}_{12}(\bb{\theta})\\
\bb{\Lambda}_{12}(\bb{\theta}) & \bb{\Lambda}_{22}(\bb{\theta})
\end{pmatrix},
\end{equation}
where $\bb{\Lambda}_{ij}(\bb{\theta})$ is a diagonal matrix whose diagonal elements are the  eigenvalues of the four $n\times n$ symmetric blocks into which the inner matrix of precision (\ref{matriceprecisione}) is partitioned. In other words, the latter are simultaneously diagonalizable through $\bb{V}_{n_1,n_2}$. This property is a consequence of the fact that
\begin{eqnarray*}
\bb{T}_{n_1,n_2}(x,y,x) & = & y\left(\bb{I}_{n_2}\otimes\bb{I}_{n_1}\right)+x\left(\bb{S}_{n_2}\otimes\,\bb{I}_{n_1}\right)+
x\left(\bb{I}_{n_2}\otimes\bb{S}_{n_1}\right),\\
\bb{S}_n & = & \tridiag(1, 0, 1),
\end{eqnarray*}
as also observed in \cite{paperunivariato} when dealing with univariate GMRFs with a second-order neighborhood. Up to a permutation of the rows and columns, the right-hand-side of the equation (\ref{caso inutile}) is a symmetric block-diagonal matrix whose diagonal blocks have size $2\times 2$. At this point, the determination of $\sigma(\bb{Q}_{n_1,n_2}(\bb{\theta}))$ is straightforward.

In the general case in which $\rho_{12}\neq\rho_{21}$, the four $n\times n$ blocks of the precision (\ref{matriceprecisione}) are not simultaneously diagonalizable. Therefore, the spectrum cannot be readily devised. The main idea to overcome this lack of analytical results is to introduce a perturbed precision matrix $\widetilde{\bb{Q}}_{n_1, n_2}(\bb{\theta})\in\R^{2n\times 2n}$, which can be partitioned into four $n\times n$ blocks which are block-circulant. This will extend the approach used in \cite{CIT-006} and is equivalent to toroidal boundary conditions on the underlying lattice. 

The perturbation is carried out as follows. Bearing in mind how the function $\bb{T}_{n_1,n_2}(\cdot)$ was defined in (\ref{funzioneT}), we see that all but its sub-blocks on the main diagonal are already circulant. On the other hand, the sub-blocks on the main diagonal are by construction tridiagonal and hence can be made circulant in a natural way as follows:
\begin{equation*}
\tridiag(x,y,z) 
\mapsto
\Circ(x,y,z) = 
\begin{pmatrix}
y & z &   &   &  x\\
x & y & z &   &  \\
  & x & \ddots & \ddots & \\
  &   & \ddots & \ddots & z \\
z  &   &        & x & y
 
\end{pmatrix}.
\end{equation*}
We explicitly point out that this method of turning a tridiagonal (more generally, non-circulant Toeplitz) matrix into a circulant one is not unique. In the literature, several approaches are discussed (see, e.g., \cite{ideaminimizzarefrobenius}). However, in our case, a simulation study provided strong empirical evidence of the fact that more sophisticated approaches are advantageous only for small choices of the grid size $n$. As we are interested in an asymptotic description of $\sigma(\bb{Q}_{n_1, n_2}(\bb{\theta}))$, we decided to stick to the natural approach. We then set:
\begin{equation}
\label{bcapproximation}
\begin{aligned}
\bb{\widetilde{Q}}_{n_1,n_2}(\bb{\theta}) & = 
\begin{pmatrix}
\bb{C}_{n_1,n_2}(\rho_{11}, 1, \rho_{11}) & \bb{C}_{n_1,n_2}(\rho_{21}, \phi, \rho_{12}) \\
\bb{C}_{n_1,n_2}\T(\rho_{21}, \phi, \rho_{12}) & \bb{C}_{n_1,n_2}(\rho_{22}, 1, \rho_{22})
\end{pmatrix}\in\R^{2n\times 2n},\\
\bb{C}_{n_1,n_2}(x,y,z) & = 
\begin{pmatrix}
\Circ(x,y,z) & \diag(z) &  &  & \diag(x)\\
\diag(x) & \Circ(x,y,z) & \diag(z) &  & \\
 & \diag(x) & \ddots & \ddots  &  \\
 &  & \ddots & \ddots & \diag(z)\\
\diag(z) &  &  & \diag(x) &  \Circ(x,y,z)
\end{pmatrix}\in\R^{n\times n}.
\end{aligned}
\end{equation}
It follows that the blocks $\bb{C}_{n_1,n_2}(\cdot)$ are block-circulant matrices of size $n$. Similarly to what was observed above for $\bb{Q}_{n_1,n_2}(\bb{\theta})$, the perturbed matrix (\ref{bcapproximation}) is block-circulant if and only if $\rho_{11} = \rho_{22}$.

Now, let $\Delta\bb{Q}_{n_1, n_2}(\bb{\theta}) = \bb{\widetilde{Q}}_{n_1, n_2}(\bb{\theta}) - \bb{{Q}}_{n_1,n_2}(\bb{\theta})$ be the perturbation matrix (see the right panel of Figure \ref{precisionepiccolaesempio} to inspect its structure). It has at most $8(n_1+n_2)\in\Theta(n_1+n_2)$ non-zero entries. This implies that the limit for $n_1,n_2\to +\infty$ of the ratio between the number of non-zero elements of $\bb{Q}_{n_1,n_2}(\bb{\theta})$ and of  $\Delta\bb{Q}_{n_1,n_2}(\bb{\theta})$ is infinite. This property will be crucial in Lemma~\ref{lemmafacile}, when we will prove our main convergence result. Apart from this, the trace of $\Delta\bb{Q}_{n_1,n_2}(\bb{\theta})$ is zero, and its rank is greater than zero, unless ${\rho_{11} = \rho_{12} = \rho_{21} = \rho_{22} = 0}$. This implies that $\Delta\bb{Q}_{n_1, n_2}(\bb{\theta})$ is an indefinite matrix; namely that it has both strictly positive and strictly negative eigenvalues. Explicit formulas for these eigenvalues can be symbolically derived, as only their algebraic multiplicity changes when $n$ grows, but they are not reported here, as they are rather complicated and not necessary for what follows. It is also important to point out that the Weyl's inequalities \citep{paperoriginaleWeyl} imply that
\begin{equation*}
\lambda_{2n}(\bb{Q}_{n_1, n_2}(\bb{\theta})) +\lambda_{2n}(\Delta\bb{Q}_{n_1, n_2}(\bb{\theta}))
\leq
\lambda_{2n}(\widetilde{\bb{Q}}_{n_1, n_2}(\bb{\theta}))
\leq
\lambda_{2n}(\bb{Q}_{n_1, n_2}(\bb{\theta})) +\lambda_{1}(\Delta\bb{Q}_{n_1, n_2}(\bb{\theta})).
\end{equation*}
These bounds are however too loose to approximate the valid parameter space (\ref{insiemevalidita}).

We now turn to the study of $\sigma(\widetilde{\bb{Q}}_{n_1,n_2}(\bb{\theta}))$, which will allow us to characterize the set~(\ref{insiemevaliditaperturbato}). In the case of a block-circulant of size $n_1\,n_2\times n_1\,n_2$ whose sub-blocks have size $n_1\times n_1$, the eigenvalues can be efficiently computed using the bidimensional FFT, with a computational complexity of $\Theta(n_1\,n_2\,\log{(n_1\,n_2)})$, as discussed in \cite{libroFFT}. For the perturbed precision (\ref{bcapproximation}), we obtain:
\begin{align}
\begin{split}
\label{autovaloriblocco12}
\sigma(\bb{C}_{n_1,n_2}(x,y,z)) = & \left\{y + z\,\exp{\left(-2\pi\I\left(\frac{i}{n_2} + \frac{j}{n_1}\right)\right)} + \right.\\
  & \left.
+x\,
\exp{\left(-2\pi\I\left(\frac{i(n_2-1)}{n_2} + \frac{j(n_1-1)}{n_1}\right)\right)}
\right\}_{i\in I,\ j\in J},
\end{split}
\end{align}
where $\I$ is the imaginary unit, $I = \{0, 1, 2,\dots, n_2-1\}$ and $J = \{0,1,2,\dots,n_1-1\}$. This result readily provides the eigenvalues of the four $n\times n$ blocks of (\ref{bcapproximation}). In particular, the eigenvalues of the two diagonal blocks are real, since the latter are by construction symmetric. For instance, for the top-left block of (\ref{bcapproximation}) it holds that
\begin{equation}
\label{autovaloriblocco11}
\sigma(\bb{C}_{n_1,n_2}(\rho_{11},1,\rho_{11})) = \left\{1+2\rho_{11}\left(\cos{\left(\frac{2\pi i}{n_2}\right)} + \cos{\left(\frac{2\pi j}{n_1}\right)} \right)\right\}_{i\in I,\ j\in J}.
\end{equation}
We are now ready to explicitly derive $\sigma(\widetilde{\bb{Q}}_{n_1,n_2}(\bb{\theta}))$. Let
\begin{equation}
\label{matriceU}
\bb{U}_{n_1,n_2} = \bb{F}_{n_1}^{(n_2)}{\bb{P}_{n_1}^{(n_2)}}\bb{F}_{n_2}^{(n_1)}\bb{P}_{n_2}^{(n_1)},
\end{equation}
where, after having set $\omega = \exp{(-2\pi\I/n)}$:
\begin{align*}
\bb{F}_n = &\, \frac{1}{\sqrt{n}}
\begin{pmatrix}
1 & 1 & 1 & \dots & 1\\
1 & \omega & \omega^2 & \dots & \omega^{n-1}\\
1 & \omega^2 & \omega^4 & \dots & \omega^{2(n-1)}\\
\vdots & \vdots & \vdots & & \vdots\\
1 & \omega^{n-1} & \omega^{2(n-1)} & \dots & \omega^{(n-1)(n-1)}
\end{pmatrix},\\
\bb{F}_{n}^{(m)} = & \,\bb{I}_{2m}\otimes\bb{F}_n,
\end{align*}
and $\bb{P}_{n}^{(m)}\in\R^{2mn\times 2mn}$ is a suitable permutation matrix. Of course, $\bb{F}_n$ is the normalized discrete Fourier transform (DFT) matrix, discussed in full detail in \cite{libroFFT}. Closed-form expression for $\bb{U}_{n_1,n_2}$ can be obtained, but once again they are not reported because they are not strictly necessary for what follows. It will suffice to observe that its non-zero entries are generally complex, with absolute value $1/{\sqrt{n_1n_2}}$. The non-zero entry pattern is displayed in Figure \ref{trasformazioneblocchetti}.
\begin{figure}
\centering
\vspace*{-6mm}
\includegraphics[scale = 0.38]{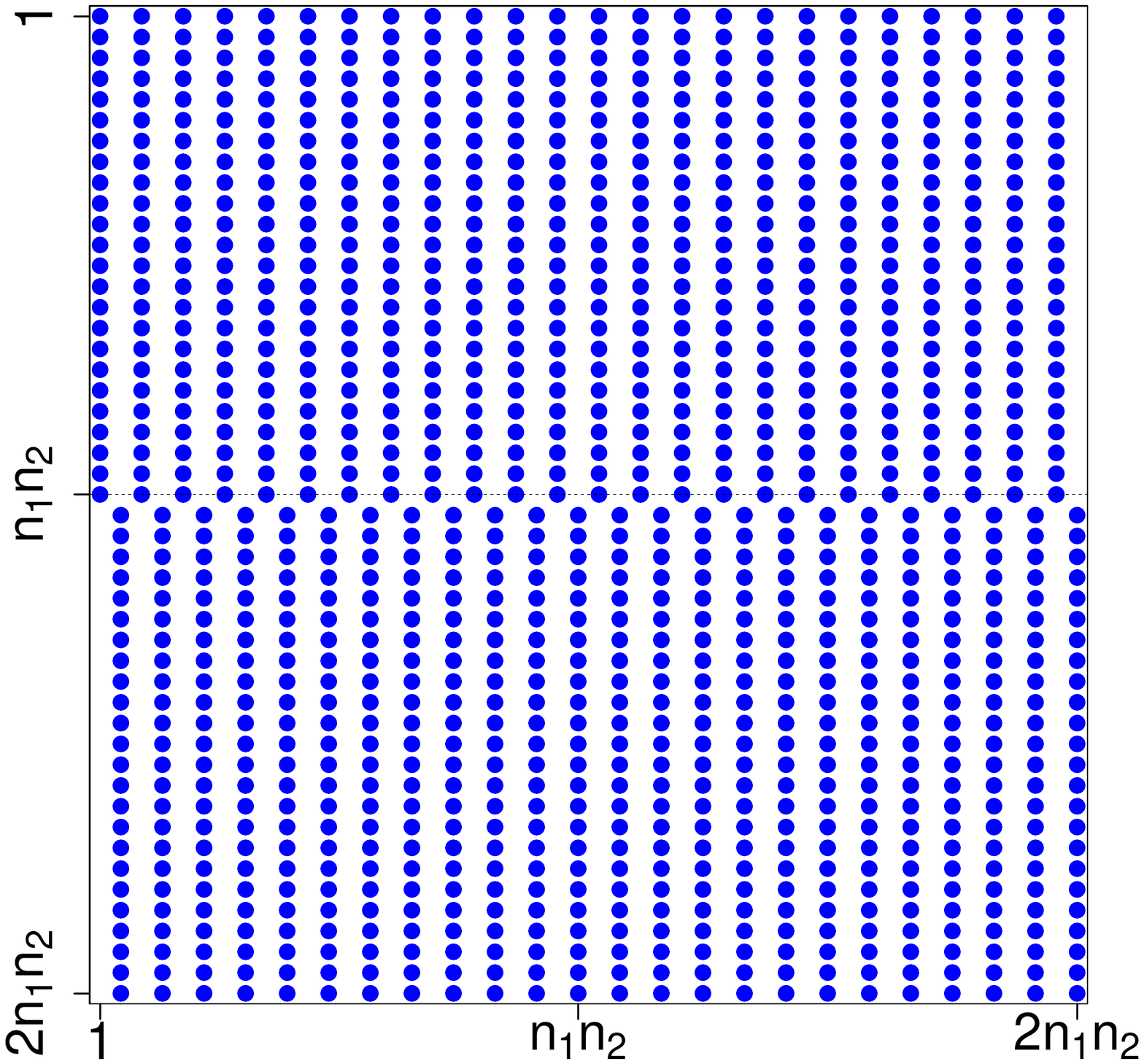}
\vspace*{-6mm}
\caption{non-zero entry pattern of $\bb{U}_{n_1,n_2}$.}
\label{trasformazioneblocchetti}
\end{figure}
In addition, $\bb{U}_{n_1,n_2}$ is by construction unitary, does not depend on $\bb{\theta}$, and is such that
\begin{equation}
\label{diagonalizzazioneblocchetti2x2}
\bb{U}_{n_1,n_2}^*\,
\widetilde{\bb{Q}}_{n_1,n_2}(\bb{\theta})\,
\bb{U}_{n_1,n_2} = 
\operatorname{blkdiag}(
\begin{pmatrix}
\widetilde{\lambda}_{i,j}^{(n)}(\rho_{11}, 1, \rho_{11}) & \widetilde{\lambda}_{i,j}^{(n)}(\rho_{21}, \phi, \rho_{12})\\
\widetilde{\lambda}_{i,j}^{(n)}(\rho_{12}, \phi, \rho_{21}) & \widetilde{\lambda}_{i,j}^{(n)}(\rho_{22}, 1, \rho_{22})
\end{pmatrix}),
\end{equation}
where $^*$ denotes the complex conjugate operator, $i\in I, j\in J$ and $\widetilde{\lambda}_{i,j}^{(n)}(x,y,z)$  is the $(i,j)$th-~eigenvalue of the block-circulant block $\bb{C}_{n_1,n_2}(x,y,z)$ of size $n$ belonging to the matrix (\ref{bcapproximation}). From the previous result, we finally obtain closed-form expressions for $\sigma(\widetilde{\bb{Q}}_{n_1,n_2}(\bb{\theta}))$, namely
\begin{equation}
\begin{multlined}
\label{autovaloribc}
\frac{1}{2}\Bigg({\widetilde{\lambda}_{i,j}^{(n)}(\rho_{11}, 1, \rho_{11})+\widetilde{\lambda}_{i,j}^{(n)}(\rho_{22}, 1, \rho_{22})}\\
\pm{\sqrt{\left(\widetilde{\lambda}_{i,j}^{(n)}(\rho_{11}, 1, \rho_{11}) - \widetilde{\lambda}_{i,j}^{(n)}(\rho_{22}, 1, \rho_{22})\right)^2+4\big\vert\widetilde{\lambda}_{i,j}^{(n)}(\rho_{21}, \phi, \rho_{12})\big\vert^2}}\Bigg),
\end{multlined}
\end{equation}
From equations (\ref{matriceU}) and (\ref{diagonalizzazioneblocchetti2x2}), several important results follow. First, a unitary matrix $\bb{V}_{n_1,n_2}(\bb{\theta})$ of size $2n\times 2n$ exists, such that the perturbed precision matrix $\widetilde{\bb{Q}}_{n_1,n_2}(\bb{\theta})$ can be diagonalized by the transformation $\bb{U}_{n_1,n_2}\bb{V}_{n_1,n_2}(\bb{\theta})$. It clearly follows that the columns of $\bb{U}_{n_1,n_2}\bb{V}_{n_1,n_2}(\bb{\theta})$ are eigenvectors of $\widetilde{\bb{Q}}_{n_1,n_2}(\bb{\theta})$. In addition, each column of $\bb{V}_{n_1,n_2}(\bb{\theta})$ contains exactly two non-zero entries. They correspond to an eigenvector of one of the $n_1n_2$ blocks of size $2\times 2$ of equation (\ref{diagonalizzazioneblocchetti2x2}). This readily provides an estimate of the absolute value of the entries of $\bb{U}_{n_1,n_2}\bb{V}_{n_1,n_2}(\bb{\theta})$. Without loss of generality, we assume that the two non-zero entries of $(\bb{V}_{n_1,n_2}(\bb{\theta}))^m$ are its two first components. Then,
\begin{eqnarray}
\nonumber
\vert(\bb{U}_{n_1,n_2}\bb{V}_{n_1,n_2}(\bb{\theta}))_l^m\vert &
= & \vert(\bb{U}_{n_1,n_2})_l\, (\bb{V}_{n_1,n_2}(\bb{\theta}))^m\vert\\
\nonumber
& = & \vert (\bb{U}_{n_1,n_2})_l^1\, (\bb{V}_{n_1,n_2}(\bb{\theta}))_1^m + (\bb{U}_{n_1,n_2})_l^2 \,(\bb{V}_{n_1,n_2}(\bb{\theta}))_2^m\vert\\
\nonumber
& \leq & \frac{1}{\sqrt{n_1n_2}}(\vert(\bb{V}_{n_1,n_2}(\bb{\theta}))_1^m \vert + \vert (\bb{V}_{n_1,n_2}(\bb{\theta}))_2^m\vert)\\
\nonumber
& \leq & \frac{\sqrt{2}}{\sqrt{n_1n_2}}\\
& \in & \Theta\left(\frac{1}{\sqrt{n_1n_2}}\right),
\label{boundDiagonalizzazionePerturbata}
\end{eqnarray}
where $l,m=1,2,\dots,2n$. The first inequality holds due to the triangular inequality and the above stated norm of the entries of $\bb{U}_{n_1,n_2}$, while the second one holds since $\bb{V}_{n_1,n_2}(\bb{\theta})$ is unitary and hence $\vert(\bb{V}_{n_1,n_2}(\bb{\theta}))_1^m \vert^2 + \vert (\bb{V}_{n_1,n_2}(\bb{\theta}))_2^m\vert^2 = 1$.  

To summarize, in this section, we solved the analogous of problem (\ref{insiemevalidita}) for the perturbed matrix~(\ref{bcapproximation}). In other words, we found closed-form expressions for the set (\ref{insiemevaliditaperturbato}). In fact, it is enough to impose that all $2n$ eigenvalues (\ref{autovaloribc}) are strictly positive. The next step is to prove asymptotically that $\bb{Q}_{n_1,n_2}(\bb{\theta})$ is positive-definite if and only if this property holds for~$\widetilde{\bb{Q}}_{n_1,n_2}(\bb{\theta})$.

We conclude this section by discussing the importance of the knowledge of the eigenvalues~(\ref{autovaloribc}) in the applications. The main result in this context is that we can test whether $\bb{\theta}$ belongs to the perturbed valid parameter space (\ref{insiemevaliditaperturbato}) by numerically checking that the eigenvalues (\ref{autovaloribc}) evaluated at $\bb{\theta}$ are positive. Once again, in the next section, we will show that this task is asymptotically equivalent to testing whether $\bb{\theta}$ belongs to (\ref{insiemevalidita}). Here, we highlight the usefulness of the closed-form expressions (\ref{autovaloribc}) in the applications from a general perspective. A uniform prior distribution on $\bb{\theta}$ has the following form:
\begin{equation*}
\pi(\bb{\theta})\propto\prod_{i = 1}^{2n} \bb{1}(\lambda_i(\bb{Q}_{n_1,n_2}(\bb{\theta})) >0 ),
\end{equation*}
where $\bb{1}(\cdot)$ is the indicator function. We need to evaluate all the $n$ terms of the form $\bb{1}(\lambda_i(\bb{Q}_{n_1,n_2}(\bb{\theta})) >0 )$ because the analytical expression of $\lambda_{2n}(\bb{Q}_{n_1,n_2}(\bb{\theta}))$ depends on complicated conditions on the components of $\bb{\theta}$. Through equation (\ref{autovaloribc}), this task can be accomplished with complexity $\Theta(n)$. The gain in computational efficiency is substantial. We benchmarked our approach against using the \texttt{R} package \texttt{spam} \cite{spam}, which provides numerical routines for sparse matrix algebra. We considered three grid sizes, namely $100\times 100$, $200\times 200$, and $300\times 300$. We observe that, if $\bb{\theta}$ induces a positive-definite precision matrix, it is then necessary to construct the latter (for instance, to evaluate the corresponding log-likelihood function), although our method in principle does not require this construction. On the contrary, if our approach determines that $\bb{\theta}$ does not belong to (\ref{insiemevalidita}), then it is not necessary to construct the corresponding precision, as the drawn $\bb{\theta}$ is to be discarded. In order to meaningfully assess both scenarios, for each of the aforementioned grid sizes, we drew 50 values of $\bb{\theta}$ belonging to the valid parameter (\ref{insiemevalidita}) and 50 values not belonging to this domain and recorded the correspondent computational times. For each of these 100 parameter values, we recorded the computational time in both \texttt{spam} and our approach. This task was accomplished through the \texttt{R} package \texttt{microbenchmark}~\cite{microbenchmark}. In the case of the 50 valid parameter values and for the three considered grid sizes, the median of the recorded computational times for \texttt{spam} was 2.6, 5.9, and 11.6 times the median obtained with our approach. The performance of our approach clearly improves when the grid size increases. In the case of the non-valid parameters, for which it was not necessary to construct the precision matrix, the median of \texttt{spam} was 40.6, 84.7, and 208.7 times the median of our approach. These timings can be further improved by exploiting the structure of the set (\ref{insiemevaliditaperturbato}), which is analytically provided by equation (\ref{autovaloribc}). It can be shown that, if we assume without loss of generality that $n_1<n_2$, the eigenvalues~(\ref{autovaloribc}) attain exactly $n_2$ local minima. It follows that we only need to evaluate $n_2$ elements of the form $\bb{1}(\lambda_i(\bb{Q}_{n_1,n_2}(\bb{\theta})) >0 )$ in order to test whether $\bb{\theta}$ belongs to the set~(\ref{insiemevaliditaperturbato}). From all the above-mentioned results, it is striking that our method outperforms the numerical routines of \texttt{spam}. This property is crucial, for instance, in an MCMC sampler, where the task of determining whether $\bb{\theta}$ is valid is usually repeated hundreds of thousands of times. Further aspects of our methodology are discussed in Section \ref{Discussion and Outlook}.
\section{The Main Result}
\label{The Main Result}
\todo{I would \emph{not} include Gray's weak convergence in this manuscript, for reasons to \textbf{be discussed}}
In the previous section, the perturbation $\widetilde{\bb{Q}}_{n_1,n_2}(\bb{\theta})$ of $\bb{Q}_{n_1,n_2}(\bb{\theta})$ was introduced. We will now formally prove our main result.
\begin{theorem}
\label{Reinhardconjecture}
Let $\bb{\theta}\in\R^5$. Then, for the precision matrices (\ref{matriceprecisione}) and (\ref{bcapproximation}), it holds that
\begin{equation*}
\lim_{n_1,n_2\to +\infty}\big\vert{\lambda_{2n}(\bb{Q}_{n_1,n_2}(\bb{\theta})) - \lambda_{2n}(\widetilde{\bb{Q}}_{n_1,n_2}(\bb{\theta}))}\big\vert  = 0.
\end{equation*}
\end{theorem} 
The previous result implies that, asymptotically, $\bb{Q}_{n_1,n_2}(\bb{\theta})$ is positive-definite if and only if this property holds for $\widetilde{\bb{Q}}_{n_1,n_2}(\bb{\theta})$.

First, we provide evidence of its novelty with respect to what is discussed in the literature in the following.
\begin{remark}
\label{remark altre convergenze}
\begin{enumerate}
\item In the literature, several convergence results are discussed for sequences of (Hermitian) matrices of the form $\{\bb{A}_n\}_{n\in\N}$ and $\{\bb{B}_n\}_{n\in\N}$, where $\bb{A}_n,\bb{B}_n\in\C^{n\times n}$. In general, they do not imply Theorem \ref{Reinhardconjecture}. For instance, in \cite{CIT-006}, the weak convergence is discussed, which is widely used in many other references, for instance in \cite{GMRFbook}. It is defined in terms of the scaled Frobenius norm of the difference $\bb{A}_n-\bb{B}_n$, which approaches zero as $n\to +\infty$, and the spectral norms of the two sequences, which have to be bounded by a constant not depending on $n$. However, this type of convergence only implies that, for large grid sizes, the spectra behave similarly as a whole. In other words, nothing can generally be stated for the single eigenvalues. For instance, we set $\bb{A}_n=\bb{I}_n$ and ${\bb{B}_n=\operatorname{diag}(1/2,1,1,\dots,1)\in\R^{n\times n}}$. It can be shown that the weak convergence holds, but clearly there is no convergence for the two sequences $\lbrace\lambda_n(\bb{A}_n)\rbrace_{n\in\N}$ and $\lbrace\lambda_n(\bb{B}_n)\rbrace_{n\in\N}$.
\item Other types of convergence were introduced in the literature, like the finite-term strong convergence discussed in \cite{convergenzafortefinita}. The goal was to overcome some of the limitations of the weak convergence in the framework of evaluating quadratic forms and log-determinants involving variance-covariance matrices with a Toeplitz structure. However, nothing can be said for the convergence of individual eigenvalues, even in this stronger framework.
\item Another classical result that, under certain technical conditions, relates the spectra of (non-circulant) Toeplitz matrices to the spectra of circulant matrices is the first Szeg\"{o} theorem, discussed in \cite{opac-b1092435}. This result was later extended in \cite{Tilli199859} to block-Toeplitz matrices. As already stated, the structure of the precision (\ref{matriceprecisione}) is more general and does not fulfill the hypotheses of the mentioned extension. 
\end{enumerate}
\end{remark}
In order to prove Theorem \ref{Reinhardconjecture}, we preliminarily recall that
\begin{eqnarray}
\label{ottimizzazione1}
\lambda_{2n}(\bb{Q}_{n_1,n_2}(\bb{\theta})) & = & \min_{\bb{x}\in\C^{2n}\mid\Vert\bb{x}\Vert = 1}\langle \bb{Q}_{n_1,n_2}(\bb{\theta})\,\bb{x},\bb{x}\rangle,\\
\label{ottimizzazione2}
\lambda_{2n}(\widetilde{\bb{Q}}_{n_1,n_2}(\bb{\theta})) & = & \min_{\bb{x}\in\C^{2n}\mid\Vert\bb{x}\Vert = 1}\langle \widetilde{\bb{Q}}_{n_1,n_2}(\bb{\theta})\,\bb{x},\bb{x}\rangle,
\end{eqnarray}
where $\langle\cdot,\cdot\rangle$ is the Euclidean inner product on $\C^{2n}$ and $\Vert\cdot\Vert$ is the correspondent induced norm. These results can be straightforwardly proved by diagonalizing the matrices ${\bb{Q}}_{n_1,n_2}(\bb{\theta})$ and $\widetilde{\bb{Q}}_{n_1,n_2}(\bb{\theta})$, respectively. In what follows, we will denote with $\bb{v}_{n_1,n_2}(\bb{\theta})$ and $\bb{u}_{n_1,n_2}(\bb{\theta})$ a solution of (\ref{ottimizzazione1}) and (\ref{ottimizzazione2}), respectively. They are eigenvectors associated with the eigenvalues $\lambda_{2n}({\bb{Q}}_{n_1,n_2}(\bb{\theta}))$ and $\lambda_{2n}(\widetilde{\bb{Q}}_{n_1,n_2}(\bb{\theta}))$, respectively. 
Based on equation (\ref{autovaloribc}) and (\ref{boundDiagonalizzazionePerturbata}), the following lemmas can be shown.
\begin{lemma}
\label{lemma ulteriore}
For any $\bb{\theta}\in\R^5$, there is a constant $C(\bb{\theta})$ such that, for any $n_1,n_2\in\N$, ${C(\bb{\theta})\leq\lambda_{2n}(\widetilde{\bb{Q}}_{n_1,n_2}(\bb{\theta}))}$ and
$$
\lim_{n_1,n_2\to +\infty} \lambda_{2n}(\widetilde{\bb{Q}}_{n_1,n_2}(\bb{\theta})) = C(\bb{\theta}).
$$
\end{lemma}
\begin{proof}
From Euler’s formula applied to equation (\ref{autovaloriblocco12}), it can be derived that
\begin{align}
\begin{split}
\label{bound norma}
\Re{\left(\widetilde{\lambda}_{i,j}^{(n)}(\bb{C}(\rho_{21}, \phi, \rho_{12}))\right)}  = &\ \phi + \rho_{12}\cos{\left(\frac{2\pi i}{n_2}\right)} + \rho_{21}\cos{\left(\frac{2\pi i}{n_2}\right)} \\
&   + \rho_{12}\cos{\left(\frac{2\pi j}{n_1}\right)} + \rho_{21}\cos{\left(\frac{2\pi j}{n_1}\right)},\\
\Im{\left(\widetilde{\lambda}_{i,j}^{(n)}(\bb{C}(\rho_{21}, \phi, \rho_{12}))\right)}  = & - \rho_{12}\sin{\left(\frac{2\pi i}{n_2}\right)} + \rho_{21}\sin{\left(\frac{2\pi i}{n_2}\right)} \\
&   - \rho_{12}\sin{\left(\frac{2\pi j}{n_1}\right)} + \rho_{21}\sin{\left(\frac{2\pi j}{n_1}\right)},
\end{split}
\end{align}
where $\Re(z)$ and $\Im(z)$ denote, respectively, the real and imaginary part of a complex number $z$. Let $\mathbb{S}$ be the unitary circle in $\C$. From the previous equalities and equation (\ref{autovaloriblocco11}), it follows that the eigenvalues (\ref{autovaloribc}) correspond to the image of the finite subset:
\begin{equation*}
\left\{\left(\sin{\left(\frac{2\pi j}{n_1}\right)},\, \cos{\left(\frac{2\pi j}{n_1}\right)}\right)\right\}_{j\in J}
\times
\left\{\left(\sin{\left(\frac{2\pi i}{n_2}\right)},\, \cos{\left(\frac{2\pi i}{n_2}\right)}\right)\right\}_{i\in I}
\end{equation*}
of $\mathbb{S}\times\mathbb{S}$ through a continuous function defined by $\mathbb{S}\times\mathbb{S}$ to $\R$ by equation (\ref{autovaloribc}) itself. We shall denote this function with $\Phi(\cdot)$. Since $\mathbb{S}\times\mathbb{S}$ is a compact subset of $\C^2$, the existence of the constant ${C}(\bb{\theta})$ follows from the Weierstrass extreme value theorem.

Let $(\bb{x}^\ast,\bb{y}^\ast)\in\mathbb{S}\times\mathbb{S}$ such that $\Phi((\bb{x}^\ast,\bb{y}^\ast)) = C(\bb{\theta})$. From equations (\ref{autovaloriblocco11}) and (\ref{bound norma}), there is a sequence $\lbrace(\bb{x}_{n_1},\bb{y}_{n_2})\rbrace_{n_1,n_2\in\N}\subset\mathbb{S}\times\mathbb{S}$ of the form
\begin{align*}
\bb{x}_{n_1}  = &\left\{\left(\sin{\left(\frac{2\pi j_{n_1}}{n_1}\right)},\, \cos{\left(\frac{2\pi j_{n_1}}{n_1}\right)}\right)\right\}_{n_1\in\N},\\
\bb{y}_{n_2}  = &\left\{\left(\sin{\left(\frac{2\pi i_{n_2}}{n_2}\right)},\, \cos{\left(\frac{2\pi i_{n_2}}{n_2}\right)}\right)\right\}_{n_2\in\N},
\end{align*}
where $i_{n_2}\in I$ and $j_{n_1}\in J$, such that $(\bb{x}_{n_1},\bb{y}_{n_2})\to(\bb{x}^\ast,\bb{y}^\ast)$.
%These sequences exist because is $\mathbb{S}$ minus a point is homeomorphic to $\left[ 0, 1\right[$ and the latter claim holds on this interval for subsets of the form $\{1/n, 2/n,\dots,(n-1)/n\}$.
Then, for continuity, $${\vert \Phi((\bb{x}_{n_1},\bb{y}_{n_2}))- C(\bb{\theta})\vert\to 0}.$$ However, $C(\bb{\theta})\leq\lambda_{2n}(\widetilde{\bb{Q}}_{n_1,n_2}(\bb{\theta})) \leq \Phi((\bb{x}_{n_1},\bb{y}_{n_2}))$, hence $${\vert C(\bb{\theta}) - \lambda_{2n}(\widetilde{\bb{Q}}_{n_1,n_2}(\bb{\theta}))\vert \leq\vert \Phi((\bb{x}_{n_1},\bb{y}_{n_2})) - C(\bb{\theta})\vert},$$ from which the second claim follows.
\end{proof}
\begin{lemma}
\label{lemmafacile}
For any $\bb{\theta}\in\R^5$, it holds that
\begin{equation*}
\lim_{n_1,n_2\to +\infty} \big\vert{\langle \bb{Q}_{n_1,n_2}(\bb{\theta})\,\bb{u}_{n_1,n_2}(\bb{\theta}),\bb{u}_{n_1,n_2}(\bb{\theta})\rangle -\lambda_{2n}(\widetilde{\bb{Q}}_{n_1,n_2}(\bb{\theta}))}\big\vert = 0.
\end{equation*}
\end{lemma}
\begin{proof}
The argument of the limit is equal to $\big\vert\langle\Delta\bb{Q}_{n_1,n_2}(\bb{\theta})\,\bb{u}_{n_1,n_2}(\bb{\theta}),\bb{u}_{n_1,n_2}(\bb{\theta})\rangle\big\vert$. Hence
\begin{eqnarray*}
\nonumber
\big\vert\langle\Delta\bb{Q}_{n_1,n_2}(\bb{\theta})\,\bb{u}_{n_1,n_2}(\bb{\theta}),\bb{u}_{n_1,n_2}(\bb{\theta})\rangle\big\vert &= &
\left\vert\sum_{l = 1}^{2n}\sum_{m = 1}^{2n}(\Delta\bb{Q}_{n_1,n_2}(\bb{\theta}))_l^m\,(\bb{u}_{n_1,n_2}(\bb{\theta}))_l\,(\bb{u}_{n_1,n_2}(\bb{\theta}))_m^*\right\vert\\
\nonumber
& \leq & \sum_{l = 1}^{2n}\sum_{m = 1}^{2n}\left\vert (\Delta\bb{Q}_{n_1,n_2}(\bb{\theta}))_l^m\right\vert\,\frac{1}{n_1n_2}\\
%\label{bound facile}
& \leq & 8\,K(\bb{\theta})\,\frac{n_1+n_2}{n_1 n_2},
\end{eqnarray*}
where $K(\bb{\theta})$ is a constant depending on $\bb{\theta}$ but not on $n$. This constant exists because the entries of $\Delta\bb{Q}_{n_1,n_2}(\bb{\theta})$ do not depend on the grid size $n$. In addition, we used the bound (\ref{boundDiagonalizzazionePerturbata}) and the fact that the latter matrix has at most $8(n_1+n_2)$ non-zero entries.
\end{proof}
In addition to the previous results, in order to prove Theorem \ref{Reinhardconjecture}, we will construct a $2\times 2$ block matrix $\bb{P}_{n_1,n_2}(\bb{\theta})\in\R^{8n\times 8n}$ with the structure that we now describe. In \cite{tesidottoratoembedding}, an approach to embed a symmetric block-Toeplitz matrix of size $
n_1n_2\times n_1n_2$ into a symmetric block-circulant matrix of size $4n_1n_2\times 4n_1n_2$ is proposed. Here, we extend this technique in a natural way to the structure of the precision matrix (\ref{matriceprecisione}). Recall that the latter is partitioned into four blocks, which are block-Toeplitz (see equation (\ref{funzioneT}) for their structure). We will start by embedding the $n_1\times n_1$ Toeplitz sub-blocks in $2n_1\times 2n_1$ circulant matrices. The sub-blocks of $\bb{T}_{n_1,n_2}(\cdot)$, which are diagonal, are embedded in a diagonal matrix of size $2n_1\times 2n_1$. In addition, we embed the $n_1\times n_1$ tridiagonal blocks in $\Circ(x,y,z)\in\R^{2n_1\times 2n_1}$. For the sake of compactness of the notation, we will denote with $\diag_{2}(x)$ the former matrix and with $\Circ_{2}(x,y,z)$ the latter. In what follows, it will be important to explicitly observe that
%\begin{equation}
\begin{align}
\begin{split}
\label{embedding blocchetti esplicito}
\diag_{2}(x)  = &
\begin{pmatrix}
\diag(x) & \bb{0}\\
\bb{0} & \diag(x)
\end{pmatrix},\\
\Circ_{2}(x,y,z)  = &
\begin{pmatrix}
\tridiag(x,y,z) & \ast\\
\ast & \tridiag(x,y,z)
\end{pmatrix}.
\end{split}
\end{align}
%\end{equation}
We define 
\begin{equation*}
\bb{P}_{n_1,n_2}(\bb{\theta}) = 
\begin{pmatrix}
\bb{E}_{n_1,n_2}(\rho_{11}, 1, \rho_{11}) & \bb{E}_{n_1,n_2}(\rho_{21}, \phi, \rho_{12})\\
\bb{E}_{n_1,n_2}(\rho_{21}, \phi, \rho_{12})\T & \bb{E}_{n_1,n_2}(\rho_{22}, 1, \rho_{22})
\end{pmatrix}\in\R^{8n\times 8n},
\end{equation*}
where $\bb{E}_{n_1,n_2}(x,y,z)\in\R^{4n\times 4n}$ is given by
\begin{equation*}
\begin{pmatrix}
\Circ_{2}(x,y,z) & \diag_{2}(z) &  &  & \diag_{2}(x)\\
\diag_{2}(x) & \Circ_{2}(x,y,z) & \diag_{2}(z) &  &  \\
 & \diag_{2}(x) & \ddots & \ddots &  \\
 &  & \ddots & \ddots & \diag_2(z)\\
\diag_{2}(z) &  &  & \diag_2(x) & \Circ_{2}(x,y,z)
\end{pmatrix}.
\end{equation*}
The construction of the matrix $\bb{E}_{n_1,n_2}(x,y,z)$  and the next lemma are better understood by means of the following
\begin{example}
Set $n_2 = 3$: then
\begin{equation*}
\bb{T}_{n_1, 3}(x,y,z) = 
\begin{pmatrix}
\tridiag(x,y,z) & \diag(z) & \bb{0} \\
\diag(x) & \tridiag(x,y,z) & \diag(z)\\
\bb{0} & \diag(x) & \tridiag(x,y,z) \\
\end{pmatrix}.
\end{equation*}
The matrix $\bb{E}_{n_1,3}(x,y,z)$ is obtained by circulation from
\begin{small}
\begin{equation*}
\left(
\begin{array}{c c | c c | c c | c c | c c | c c}
\tridiag(x,y,z) & \ast & \diag(z) & \bb{0} & \bb{0} & \bb{0} & \bb{0} & \bb{0} & \bb{0} & \bb{0} & \diag(x) & \bb{0}\\
\ast & \tridiag(x,y,z) & \bb{0} & \diag(z) & \bb{0} & \bb{0} & \bb{0} & \bb{0} & \bb{0} & \bb{0} & \bb{0} & \diag(x)\\
\hline
\vdots & \vdots &\vdots &\vdots &\vdots &\vdots &\vdots &\vdots &\vdots &\vdots &\vdots &\vdots 
\end{array}
\right).
\end{equation*}
\end{small}
In general, in the expression above, there are $2n_2-3$ blocks of the form $\bb{0}_{2n_1\times 2n_1}$. 
\end{example}
We prove the main properties of the matrix $\bb{P}_{n_1,n_2}(\bb{\theta})$ in the following
\begin{lemma}
\label{lemmaP}
Let $\bb{\theta}\in\R^5$. Then,
\begin{enumerate}
\item $\lambda_{8n}(\bb{P}_{n_1,n_2}(\bb{\theta})) \leq \lambda_{2n}({\bb{Q}}_{n_1,n_2}(\bb{\theta}))$
\item $\lim_{n_1,n_2\to +\infty} \vert \lambda_{2n}(\bb{\widetilde{Q}}_{n_1,n_2}(\bb{\theta}))-\lambda_{8n}(\bb{P}_{n_1,n_2}(\bb{\theta}))\vert = 0$
\end{enumerate}
\end{lemma}
\begin{proof}
\begin{enumerate}
\item Up to a permutation of the rows and columns of $\bb{P}_{n_1,n_2}(\bb{\theta})$, it holds that
\begin{equation}
\label{permutazionefurba}
\bb{P}_{n_1,n_2}(\bb{\theta}) =
\begin{pmatrix}
\bb{Q}_{n_1,n_2}(\bb{\theta}) & \ast\\
\ast & \ast
\end{pmatrix}.
\end{equation}
In order to show this claim, we introduce the set of indices given by
\begin{equation*}
\mathcal{I} = \{1,2,\dots,n_1,\,2n_1+1,\dots,3n_1,\,\dots\dots,\,2(n_2-1)n_1+1,\dots,(2n_2-1)n_1 \},
\end{equation*}
which has cardinality $n_1n_2$. Then $\bb{T}_{n_1,n_2}(\rho_{11}, 1, \rho_{11}) = (\bb{E}_{n_1,n_2}(\rho_{11}, 1, \rho_{11}))_i^i$, with $i\in\mathcal{I}$, as a consequence of (\ref{funzioneT}) and (\ref{embedding blocchetti esplicito}). This line of reasoning, when applied block-wise to $\bb{P}_{n_1,n_2}(\bb{\theta})$, implies (\ref{permutazionefurba}). The well-known Cauchy interlacing theorem, discussed in \cite{InterlacingInequalities}, yields the inequality $\lambda_{8n}(\bb{P}_{n_1,n_2}(\bb{\theta}))\leq \lambda_{2n}(\bb{Q}_{n_1,n_2}(\bb{\theta}))$.

\item The matrix $\bb{P}_{n_1,n_2}(\bb{\theta})$ is partitioned into four blocks $\bb{E}_{n_1,n_2}(\cdot)$, which are of the form $\bb{C}_{n_1,n_2}(\cdot)$, as defined in equation (\ref{bcapproximation}). Therefore, the eigenvalues of $\bb{P}_{n_1,n_2}(\bb{\theta})$ can be obtained from equation (\ref{autovaloriblocco12}) by considering a regular lattice of size $2n_1\times 2n_2$. Hence, $\lbrace\lambda_{8n}(\bb{P}_{n_1,n_2}(\bb{\theta}))\rbrace_{n\in\N}$ is a sub-sequence of $\lbrace\lambda_{2n}(\widetilde{\bb{Q}}_{n_1,n_2}(\bb{\theta}))\rbrace_{n\in\N}$. For Lemma \ref{lemma ulteriore}, the latter sequence is convergent, hence it is a Cauchy sequence.
\end{enumerate}
\end{proof}
Now, we have all the results that we need to show Theorem \ref{Reinhardconjecture}.
\begin{proof}[Proof of Theorem~\ref{Reinhardconjecture}]
It holds that
\begin{equation*}
\lambda_{2n}(\bb{Q}_{n_1,n_2}(\bb{\theta})) = \langle \bb{Q}_{n_1,n_2}(\bb{\theta})\,\bb{v}_{n_1,n_2}(\bb{\theta}),\bb{v}_{n_1,n_2}(\bb{\theta})\rangle
\leq
\langle \bb{Q}_{n_1,n_2}(\bb{\theta})\,\bb{u}_{n_1,n_2}(\bb{\theta}),\bb{u}_{n_1,n_2}(\bb{\theta})\rangle,
\end{equation*}
as $\bb{v}_{n_1,n_2}(\bb{\theta})$ is a solution of the variational problem (\ref{ottimizzazione1}). In addition:
\begin{equation}
\label{catenageniale}
\lambda_{2n}(\bb{\widetilde{Q}}_{n_1,n_2}(\bb{\theta}))\sim\lambda_{8n}(\bb{P}_{n_1,n_2}(\bb{\theta}))\leq
\lambda_{2n}(\bb{Q}_{n_1,n_2}(\bb{\theta}))\leq \langle\bb{Q}_{n_1,n_2}(\bb{\theta})\,\bb{u}_{n_1,n_2}(\bb{\theta}),\bb{u}_{n_1,n_2}(\bb{\theta})\rangle,
\end{equation}
where for $\{ a_n\}_{n\in\N}$ and $\{b_n\}_{n\in\N}$ we write, with slight abuse of notation, $a_n\sim b_n$ if and only if $\lim_{n\to +\infty}\vert a_n-b_n\vert = 0$. Now, $\lambda_{2n}(\bb{\widetilde{Q}}_{n_1,n_2}(\bb{\theta}))\sim\lambda_{8n}(\bb{P}_{n_1,n_2}(\bb{\theta}))$ is the second claim of Lemma~\ref{lemmaP}, while $\lambda_{8n}(\bb{P}_{n_1,n_2}(\bb{\theta}))\leq
\lambda_{2n}(\bb{Q}_{n_1,n_2}(\bb{\theta}))$ is the first claim thereof. The chain (\ref{catenageniale}) is therefore proved by applying Lemma \ref{lemmafacile}. It implies Theorem \ref{Reinhardconjecture}.
\end{proof}
\begin{remark}
\label{remarkfalsepositives}
The chain (\ref{catenageniale}) does not imply that $\lambda_{2n}(\bb{\widetilde{Q}}_{n_1,n_2}(\bb{\theta}))\leq \lambda_{2n}(\bb{{Q}}_{n_1,n_2}(\bb{\theta}))$. If it would have been possible to theoretically show this inequality, then we would not have needed to construct the matrix $\bb{P}_{n_1,n_2}(\bb{\theta})$. On the other hand, this property is beneficial in the applications. If the perturbed precision (\ref{bcapproximation}) is positive-definite, then the precision (\ref{matriceprecisione}) has the same property. This avoids the case of \lq\lq false positives,\rq\rq namely efficiently sampling from (\ref{insiemevaliditaperturbato}) values of $\bb{\theta}$ which do not belong to the domain (\ref{insiemevalidita}).
\end{remark}
From the previous proof, only a loose bound of the rate of convergence of \mbox{Theorem  \ref{Reinhardconjecture}} can be found. This issue stems from the fact we had to consider the auxiliary quantities $ \langle\bb{Q}_{n_1,n_2}(\bb{\theta})\,\bb{u}_{n_1,n_2}(\bb{\theta}),\bb{u}_{n_1,n_2}(\bb{\theta})\rangle$ and $\lambda_{8n}(\bb{P}_{n_1,n_2}(\bb{\theta}))$ in order to avoid the lack of analytical results on the spectrum of ${\bb{Q}}_{n_1,n_2}(\bb{\theta})$. As a consequence, we can only get loose bounds on $\vert \lambda_{2n}(\bb{{Q}}_{n_1,n_2}(\bb{\theta})) - \lambda_{2n}(\bb{\widetilde{Q}}_{n_1,n_2}(\bb{\theta}))\vert$ in terms of the just mentioned two quantities. %More precisely, chain (\ref{catenageniale}) only implies that
%\begin{align*}
%\vert \lambda_{2n}(\bb{\widetilde{Q}}_{n_1,n_2}(\bb{\theta})) -  \lambda_{2n}(\bb{{Q}}_{n_1,n_2}(\bb{\theta}))\vert\leq & 
%\big\vert \lambda_{8n}(\bb{P}_{n_1,n_2}(\bb{\theta}))  
%- \langle\bb{Q}_{n_1,n_2}(\bb{\theta})\,\bb{u}_{n_1,n_2}(\bb{\theta}),\bb{u}_{n_1,n_2}(\bb{\theta})\rangle
%\big\vert  \\  + & \big\vert\lambda_{8n}(\bb{P}_{n_1,n_2}(\bb{\theta})) - \lambda_{2n}(\bb{\widetilde{Q}}_{n_1,n_2}(\bb{\theta}))  \big\vert.
%\end{align*}
In the next section, we will provide a thorough numerical study aimed at empirically obtaining a tight upper bound of the rate of convergence.
\section{Assessment of the Rate of Convergence}
\label{Assessment of the Rate of Convergence}
In this section, we numerically assess the rate of convergence in Theorem \ref{Reinhardconjecture}. As a preliminary step, we briefly inspect the asymptotic behavior in both the case of a univariate GMRF over a transect of size $n$, namely a grid constructed on a line, and over a bidimensional regular lattice of size $n_1\times n_2$. This will be helpful later to interpret the more general case of approximating the precision (\ref{matriceprecisione}) with the matrix (\ref{bcapproximation}). For a univariate GMRF, the eigenvalues of the associated precision matrix can be analytically determined. On the contrary, in the bivariate case, we will need to resort to numerical techniques, since the eigenvalues of the precision (\ref{matriceprecisione}) are not known, as already stated. We will keep the notation as close as possible to the one that we used when we introduced precision (\ref{matriceprecisione}). Since we will be concerned with matrices whose structure is similar to or the same as the one of the $n_1n_2\times n_1n_2$ diagonal blocks of (\ref{matriceprecisione}), we will therefore denote the off-diagonal non-zero elements with $\rho$. 

For a univariate field over a transect of length $n$, the Toeplitz precision $\tridiag(\rho,1,\rho)$ is naturally approximated by $\Circ(\rho,1,\rho)$. The eigenvalues of a Toeplitz tridiagonal matrix are well-known (see, e.g., \cite{NLA1811}). It holds that ${\lambda_n(\tridiag(\rho,1,\rho)) = 1 - 2\vert \rho\vert\,\cos{(\pi/(n+1))}}$. In addition, if $\rho < 0$, then ${\lambda_n(\Circ(\rho,1,\rho)) =  1 - 2\vert \rho\vert}$, whereas, if $\rho > 0$:
\begin{equation*}
\lambda_n(\Circ(\rho,1,\rho)) =
\begin{cases}
1 - 2\vert \rho\vert, & \text{if $n$ is even},\\
1 + 2\vert \rho\vert \,\cos{\left(\frac{2\pi}{n}\lfloor \frac{n}{2}\rfloor\right)}, & \text{if $n$ is odd}.
\end{cases}
\end{equation*}
We consider the absolute error $\epsilon_n(\rho) =\vert \lambda_n(\tridiag(\rho,1,\rho)) - \lambda_n(\Circ(\rho,1,\rho)\vert$. Then, if $\rho < 0$, $1-2\vert\rho\vert=\lambda_n(\Circ(\rho, 1, \rho))\leq\lambda_n(\tridiag(\rho, 1,\rho))$ and $\epsilon_n(\rho)\in\Theta(1/n^2)$. This can be proved by expanding the function $n\mapsto\cos{(\pi/(n+1))}$ in the Taylor series evaluated at $n = \infty$. Otherwise, if $\rho > 0$, the same order of convergence can be obtained if $n$ is even. If $n$ is odd, $\epsilon_n(\rho)\in\Theta(1/n^3)$, but $\lambda_n(\tridiag(\rho, 1,\rho))\leq\lambda_n(\Circ(\rho, 1, \rho))$. These different rates of convergence, which depend on the parity of $n$, will allow us to better interpret the convergence pattern in the more general case of a bivariate field over a regular lattice.

We now consider the case of a univariate GMRF over a bidimensional regular lattice of size $n = n_1\times n_2$, whose precision matrix is $\bb{T}_{n_1,n_2}(\rho, 1,\rho)$. The eigenvalues of this matrix are explicitly derived in \cite{paperunivariato}. Based on our approach, $\bb{T}_{n_1,n_2}(\rho, 1,\rho)$ is approximated with $\bb{C}_{n_1,n_2}(\rho, 1,\rho)$. Similarly to what was done above, we consider the absolute error $\epsilon_n(\rho) = \vert\lambda_{n}(\bb{T}_{n_1,n_2}(\rho, 1,\rho)) - \lambda_{n}(\bb{C}_{n_1,n_2}(\rho, 1,\rho))\vert$. Using closed-form expressions (\ref{autovaloribc}), it can be shown that, if $\rho < 0$, then ${\epsilon_n \in \Theta({1}/{n})}$, regardless of the parity of $n_1$ and $n_2$. If $\rho > 0$, we need to distinguish between the following three cases:
\begin{itemize}
\item If $n_1$ and $n_2$ are even, then ${\epsilon_n \in \Theta({1}/{n})}$ and $\lambda_{n}(\bb{C}_{n_1,n_2}(\rho, 1,\rho))\leq\lambda_{n}(\bb{T}_{n_1,n_2}(\rho, 1,\rho))$. In addition, $\lambda_{n}(\bb{C}_{n_1,n_2}(\rho, 1,\rho)) = C(\rho)$, where $C(\rho)$ is a constant not depending on~$n$.
\item If exactly one of $n_1$ or $n_2$ is even, then the same rate of convergence of the previous case holds and $\lambda_{n}(\bb{C}_{n_1,n_2}(\rho, 1,\rho))\leq\lambda_{n}(\bb{T}_{n_1,n_2}(\rho, 1,\rho))$,  but $\{\lambda_{n}(\bb{C}_{n_1,n_2}(\rho, 1,\rho))\}_{n\in\N}$ is not a constant sequence.
\item If $n_1$ and $n_2$ are odd, then ${\epsilon_n \in \Theta({1}/{n^{1.5}})}$, but, in contrast to the previous two cases, $\lambda_{n}(\bb{C}_{n_1,n_2}(\rho, 1,\rho))\geq\lambda_{n}(\bb{T}_{n_1,n_2}(\rho, 1,\rho))$.
\end{itemize}
The proof of the latter three claims is analogous to the ones of the transect case. For what follows, it is important to note that the asymptotic behavior of the error depends on the parity of $n_1$ and $n_2$ and on the sign of $\rho$. We also point out that similar results are discussed in \citep{GMRFbook} with respect to the maximum likelihood estimator of (univariate) GMRFs and the dimension of the associated lattice.

We now return to the more general case of approximating the precision (\ref{matriceprecisione}) with the matrix (\ref{bcapproximation}) in the bivariate case over a regular lattice. We will provide numerical evidence of the fact that the rate of convergence in Theorem \ref{Reinhardconjecture} is in the form $\Theta(1/n)$. Similarly to what was observed for a univariate field over a regular lattice, the pattern of convergence of the minimum eigenvalue of the perturbed precision matrix (\ref{bcapproximation}) to the minimum eigenvalue of the precision (\ref{matriceprecisione}) will depend on both the parity of $n_1$ and $n_2$ and on the value of $\bb{\theta}$. Due to the higher dimension of the set (\ref{insiemevalidita}), the fluctuation pattern as $n$ grows will require a more thorough study. Apart from this, we recall that in Remark \ref{remarkfalsepositives} we briefly discussed the importance in the applications of the absence of \lq\lq false positives.\rq\rq\ In what follows, we will also take this property into account.

We consider the following numerical experiment. First, we uniformly sample $\bb{\theta}_1,\bb{\theta}_2,\dots,\bb{\theta}_N$ from $[-1, 1]^5$, which can be shown to be a superset of (\ref{insiemevalidita}) for any $n_1,n_2\in\N$. Second, for each of these draws, we consider increasing grid sizes, and we numerically determine the minimum eigenvalues of the perturbed precision (\ref{bcapproximation}) and of the precision matrix for each of them (\ref{matriceprecisione}). For the former eigenvalue, we use the closed-form expressions (\ref{autovaloribc}). For the latter eigenvalue, we use the \texttt{eigs} command of the \texttt{R} package \texttt{rARPACK} \cite{bruttacopiamatlab}, which implements the Lanczos algorithm \cite{Lanczos50aniterative} to determine the extreme eigenvalues (and possibly eigenvectors) of a square complex matrix. Third, we numerically compute the constant $C(\bb{\theta})$ introduced in Lemma \ref{lemma ulteriore} and focus on the quantities 
\begin{eqnarray*}
\epsilon_{n_1,n_2}(\bb{\theta}) & = & \big\vert \lambda_{2n}(\bb{\widetilde{Q}}_{n_1,n_2}(\bb{\theta})) - \lambda_{2n}(\bb{{Q}}_{n_1,n_2}(\bb{\theta}))\big\vert,\\
\delta_{n_1,n_2}(\bb{\theta}) & = & \vert \lambda_{2n}(\bb{{Q}}_{n_1,n_2}(\bb{\theta})) - C(\bb{\theta})\vert.
\end{eqnarray*}
The second quantity was introduced because we will provide numerical evidence of the fact that $\lambda_{2n}(\bb{\widetilde{Q}}_{n_1,n_2}(\bb{\theta}))\leq\lambda_{2n}(\bb{{Q}}_{n_1,n_2}(\bb{\theta}))$ definitely, hence $\epsilon_{n_1,n_2}(\bb{\theta})\leq\delta_{n_1,n_2}(\bb{\theta})$ definitely due to chain (\ref{catenageniale}). Moreover, $\delta_{n_1,n_2}(\bb{\theta})$ exhibits an asymptotic convergence pattern that does not depend on the parity of $n_1,n_2$ or on the value of $\bb{\theta}$, in contrast to $\epsilon_{n_1,n_2}(\bb{\theta})$. 

In the following, we set $N = 500$ and considered the grid sizes $50\times 50$ to $300\times 300$, for a total of 251 increasing values. For each of the 500 drawn values of $\bb{\theta}$ we fitted a least-squares line of $\log_{10}{(\delta_{n_1,n_2}(\bb{\theta}))}$ against $\log_{10}{(n_1n_2)}$ (recall that the precision matrix (\ref{matriceprecisione}) has size $2n_1 n_2\times 2n_1 n_2$). As a measure of the quality of the fit, we considered the (non-adjusted) ~$R^2$. The minimum observed $R^2$ was 0.999, which provides very strong evidence for a rate of convergence of $\delta_{n_1,n_2}(\bb{\theta})$ in the form $\Theta(1/n^\alpha)$, without any fluctuations. Furthermore, since the empirical distribution of the $500$ estimated slopes had median $-0.992$ and the first and third quartiles were $-1.002$ and $-0.989$, respectively, we empirically conclude that $\delta_{n_1,n_2}(\bb{\theta})\in\Theta(1/n)$ or, equivalently, that $\alpha = 1$. 

The convergence pattern of $\lambda_n(\widetilde{\bb{Q}}_{n_1,n_2}(\bb{\theta}))$ to $\lambda_n({\bb{Q}}_{n_1,n_2}(\bb{\theta}))$ is more complicated and depends on the value of $\bb{\theta}$, and on the parity of $n_1$ and $n_2$. In Figure \ref{differentialplot}, three values thereof are displayed.
\begin{figure}
\centering
\vspace*{-6pt}
\includegraphics[scale = 0.69]{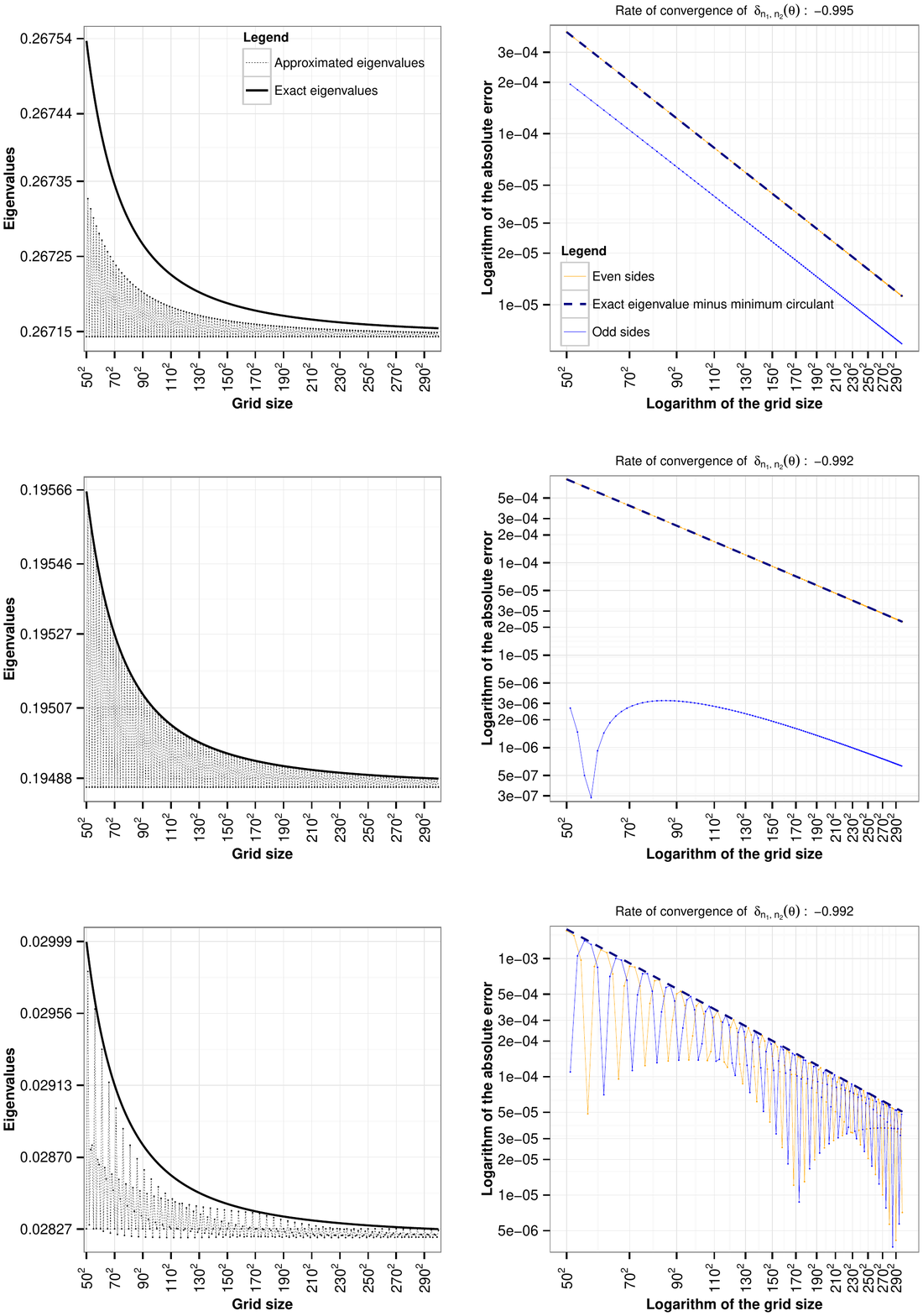}
\caption{in the left column, the convergence pattern of $\lambda_{n}(\widetilde{\bb{Q}}_{n_1,n_2}(\bb{\theta}))$ to $\lambda_{n}({\bb{Q}}_{n_1,n_2}(\bb{\theta}))$ against increasing grid sizes is displayed for three different values of $\bb{\theta}$. In the right column, the correspondent errors $\epsilon_{n_1, n_2}(\cdot)$, and $\delta_{n_1,n_2}(\cdot)$ are plotted on a log-log scale.}
\label{differentialplot}
\end{figure}
They were chosen in order to provide a meaningful overview of all observed convergence patterns. The left column shows the obtained results on the usual scale, while the right column displays the convergence of $\delta_{n_1,n_2}(\bb{\theta})$ and $\epsilon_{n_1,n_2}(\bb{\theta})$ to zero on a log-log scale. As stated above, the convergence pattern of the former quantity is a straight line of slope $-1$, whereas the pattern of the latter depends on both the parity of $n_1,n_2$ and the value of $\bb{\theta}$. This fact is very similar to what has already been observed in the beginning of this section for univariate fields over a transect and a regular lattice. More precisely, with respect to the first considered value of $\bb{\theta}$, it holds that ${\lambda_{n}(\widetilde{\bb{Q}}_{n_1,n_2}(\bb{\theta})) = C(\bb{\theta})}$ if and only if $n_1$ and $n_2$ are even, whereas $\lambda_{n}(\widetilde{\bb{Q}}_{n_1,n_2}(\bb{\theta}))\to C(\bb{\theta})$, but the sequence is not constant, if and only if $n_1$ and $n_2$ are odd. In addition, the absolute error $\epsilon_{n_1,n_2}(\bb{\theta})$ corresponds to a line with slope $-1$ in the log-log scale, with the intercept dependent on the parity of $n_1$ and $n_2$, and $\lambda_{n}(\widetilde{\bb{Q}}_{n_1,n_2}(\bb{\theta}))\leq \lambda_{n}({\bb{Q}}_{n_1,n_2}(\bb{\theta}))$ for every considered grid size. The second considered value of $\bb{\theta}$ exhibits a similar behavior for grid sizes with both $n_1$ and $n_2$ even. The spike that can be seen for odd values of $n_1$ and $n_2$ corresponds to the grid size at which the sign of ${\lambda_{n}(\widetilde{\bb{Q}}_{n_1,n_2}(\bb{\theta}))-\lambda_{n}({\bb{Q}}_{n_1,n_2}(\bb{\theta}))}$ changes. For larger grid-sides, it holds that ${\lambda_{n}(\widetilde{\bb{Q}}_{n_1,n_2}(\bb{\theta}))\leq \lambda_{n}({\bb{Q}}_{n_1,n_2}(\bb{\theta}))}$. Finally, the third considered value of $\bb{\theta}$ shows the most complicated convergence pattern. From the right panel, it is clear that the convergence patterns of grid sizes with even and odd values of $n_1$ and $n_2$ are similar. Differently from the above discussed cases, $\epsilon_{n_1,n_2}(\bb{\theta})$ does not (asymptotically) correspond to a line. In addition, it is empirically clear from the left panel that $\lambda_{n}(\widetilde{\bb{Q}}_{n_1,n_2}(\bb{\theta}))\leq \lambda_{n}({\bb{Q}}_{n_1,n_2}(\bb{\theta}))$, hence $\epsilon_{n_1,n_2}(\bb{\theta})\leq\delta_{n_1,n_2}(\bb{\theta})$. In other words, regardless of the complex fluctuation pattern, the absolute error is bounded from above by $\delta_{n_1,n_2}(\bb{\theta})$, which has the nice convergence pattern discussed above.

To summarize, in this section, we provided strong empirical evidence of the fact that the rate of convergence in Theorem \ref{Reinhardconjecture} is $\Theta(1/n)$. In addition, it holds that, asymptotically, ${\lambda_{n}(\widetilde{\bb{Q}}_{n_1,n_2}(\bb{\theta}))\leq \lambda_{n}({\bb{Q}}_{n_1,n_2}(\bb{\theta}))}$. As discussed in Remark~\ref{remarkfalsepositives}, this property implies that, for large grid sizes, our methodology does not lead to \lq\lq false positives.\rq\rq

\clearpage
\section{Discussion and Outlook}
\label{Discussion and Outlook}
In this paper, asymptotically closed-form expressions for the determination of the valid parameter space (\ref{insiemevalidita}) through a suitable approximation (\ref{insiemevaliditaperturbato}) were provided. The importance of this result in the applications is twofold. First, it provides an efficient tool in order to efficiently sample from (\ref{insiemevalidita}). This task is crucial, for instance, in the framework of a Metropolis–Hastings within the Gibbs sampler aimed at analyzing multivariate and highly correlated data with the model introduced in \cite{Sain:Furr:Cres:11} or extensions thereof. Second, current work involves the generalization of the convergence results discussed in \cite{CIT-006} by means of the eigenvalues (\ref{autovaloribc}) and Theorem~\ref{Reinhardconjecture}. The goal is to efficiently evaluate quadratic forms and log-determinants involving matrices of the form (\ref{matriceprecisione}) without the need for computing the Cholesky decomposition of large sparse matrices, for example with the routines provided by \texttt{spam} \cite{spam}. This is a very demanding task from a computational point view. The need for determining the Cholesky factor when sampling from the latent field $\bb{z}$ of the form (\ref{secondo layer}) can also be avoided by exploiting the results discussed in \cite{Aune2013}.

The rate of convergence in Theorem \ref{Reinhardconjecture} was numerically assessed to be $\Theta(1/n)$. Moreover, we provided evidence of the fact that, asymptotically, $\lambda_{2n}(\bb{\widetilde{Q}}_{n_1,n_2}(\bb{\theta}))\leq
\lambda_{2n}(\bb{Q}_{n_1,n_2}(\bb{\theta}))$. In the applications, this avoids \lq\lq false positives\rq\rq\ in the approximation of (\ref{insiemevalidita}) with (\ref{insiemevaliditaperturbato}). In fact, if $\lambda_{2n}(\widetilde{\bb{Q}}_{n_1,n_2}(\bb{\theta}))>0$, namely $  \widetilde{\bb{Q}}_{n_1,n_2}(\bb{\theta})$ is positive-definite, then $\lambda_{2n}({\bb{Q}}_{n_1,n_2}(\bb{\theta}))>0$; hence, ${\bb{Q}}_{n_1,n_2}(\bb{\theta})$ is also positive-definite.

The closed-form expressions (\ref{autovaloribc}) allow us to inspect the geometrical properties of the parameter space (\ref{insiemevalidita}) (e.g., connection and compactness), which play a central when implementing likelihood function optimizers as discussed, e.g., in \citep{bb47473, 54906}. Some preliminary insight of the shape of the set (\ref{insiemevalidita}) is already provided in Figure \ref{figuradd}. Several theoretical results on the geometrical structure of the valid parameter space (\ref{insiemevalidita}) were also obtained, but they are not included in this paper because they are beyond its scope.

From a statistical perspective, the model introduced in \cite{Sain:Furr:Cres:11} is such that the correlation structure of different locations on the lattice depends only on the parameters $\rho_{11},\rho_{22},\rho_{12}$, and $\rho_{21}$. Our knowledge of the valid parameter space (\ref{insiemevalidita}) will allow us to add a further layer to the above-mentioned model, which is in the form (\ref{modellogenerale}). For instance, $\bb{\theta}$ could be written as a suitable function of latitude and longitude. This function must ensure that the precision matrix (\ref{matriceprecisione}) is positive-definite, hence the importance of the knowledge of the eigenvalues (\ref{autovaloribc}).

The strategy used in Section \ref{The Main Result} to prove the main Theorem~\ref{Reinhardconjecture}, namely the introduction of the chain (\ref{catenageniale}), will be used to prove similar results for the more general case of $p$-variate GMRFs. 
The associated precision matrix can be partitioned in $p^2$ block-Toeplitz blocks of size ${n_1n_2\times n_1n_2}$, but, depending on the parametrization of the GMRF, it may not be block-Toeplitz. Once again, we point out that no extensions of the classical convergence results described in Remark~\ref{remark altre convergenze} are (to our knowledge) available for this framework. We will therefore apply a perturbation similar to the one introduced in equation (\ref{bcapproximation}) and then obtain a closed-form expression similar to the eigenvalues (\ref{autovaloribc}). Apart from this, another interesting research question to address would be to explore the applicability of our approach for neighborhood structures of higher order in the lattice.

\section*{Acknowledgments}
This work was supported by the Swiss National Science Foundation, grant 143282. The first author would like to thank Antonio De Rosa, Giuseppe Graziani, and Salvatore Stuvard for the helpful discussions while developing the proof of the main result.

The content of this paper is part of the first author’s Ph.D. dissertation supervised by the second author.

\renewcommand{\baselinestretch}{1.}\rm

\bibliographystyle{as}
\bibliography{mybib}
\end{document}